%% file: example_paper.tex
\documentclass{article}

\usepackage{microtype}
\usepackage{graphicx}
\usepackage{subfigure}
\usepackage{booktabs} 
\usepackage{comment}
\usepackage{tikz} 
\usetikzlibrary{arrows}
\usepackage{hyperref}

\usepackage[accepted]{icml2025}
\input{math}

\usepackage{amsmath}
\usepackage{amssymb}
\usepackage{mathtools}
\usepackage{amsthm}

\usepackage[capitalize,noabbrev]{cleveref}

\theoremstyle{plain}
\newtheorem{theorem}{Theorem}[section]
\newtheorem{proposition}[theorem]{Proposition}
\newtheorem{lemma}[theorem]{Lemma}

\theoremstyle{definition}
\newtheorem{definition}[theorem]{Definition}

\theoremstyle{remark}

\usepackage[textsize=tiny]{todonotes}

\icmltitlerunning{Towards Robust Causal Effect Identification Beyond Markov Equivalence}

\begin{document}

\twocolumn[
\icmltitle{Towards Robust Causal Effect Identification Beyond Markov Equivalence}



\icmlsetsymbol{equal}{*}

\begin{icmlauthorlist}
\icmlauthor{Kai Z. Teh}{yyy}
\icmlauthor{Kayvan Sadeghi}{yyy}
\icmlauthor{Terry Soo}{yyy}
\end{icmlauthorlist}

\icmlaffiliation{yyy}{Department of Statistical Science, University College London, London, UK}

\icmlcorrespondingauthor{Kai Teh}{kai.teh.21@ucl.ac.uk}

\icmlkeywords{Machine Learning, ICML}

\vskip 0.3in
]



\printAffiliationsAndNotice{\icmlEqualContribution} 

\begin{abstract}
Causal effect identification typically requires a fully specified causal graph, which can be difficult to obtain in practice. We provide a sufficient criterion for identifying causal effects from a candidate set of Markov equivalence classes with added background knowledge, which represents cases where determining the causal graph up to a single Markov equivalence class is challenging. Such cases can happen, for example, when the untestable assumptions (e.g. faithfulness) that underlie causal discovery algorithms do not hold.

\end{abstract}

\section{Introduction}\label{sec:intro}
Identifying causal effects from observational data is important when intervention experiments cannot be performed. Under causal sufficiency, observational data can be represented using a causal graph in the form of a directed acyclic graph (DAG). If the causal DAG is known, then all causal effects can be identified from observational data.

However, the underlying causal DAG is usually unknown and has to be obtained from observational data using causal discovery algorithms, which can at most learn the causal DAG up to its Markov equivalence class, represented as a completed partially directed acyclic graph (CPDAG) \citep{meek}. As such, works regarding causal effect identification for CPDAGs \citep{cpdagidentify}, as well as for maximally oriented partially directed acyclic graphs (MPDAGs), which represent CPDAGs with background knowledge constraints, have been proposed \citep{perk}. These methods assume that the Markov equivalence class of the true causal DAG is known.

However, causal discovery algorithms often assume untestable assumptions, such as faithfulness, which can often be too strong in practice. When these untestable assumptions are violated, the Markov equivalence class returned by causal discovery algorithms need not be unique and can instead be a candidate set of Markov equivalence classes \citep{teh1}.

As an example, let the observational distribution \(P\) (over \(X_i\)) be generated via
\begin{align*}
  \epsilon_i, \phi_j \stackrel{\text{i.i.d.}}{\sim} Bern(\tfrac{1}{2}), i=1,\ldots ,4, j=1,\ldots,5,\nonumber\\
  X_1 = (\phi_1, \phi_2, \epsilon_1),\nonumber\\
  X_2=  (X^1_1,\phi_3, \epsilon_2),\nonumber\\
  X_3=  (\phi_4,\phi_5, \epsilon_3),\nonumber\\
  X_4=   (X_1^1+X_3^1, X_2^1+X_3^2, X_2^2, \epsilon_4),\label{sem1}
\end{align*}
where 
\(X^j_i\) denotes the \(j\)-th entry from the left of \(X_i\).

Both CPDAGs \(\mathcal{G}_1\) and \(\mathcal{G}_2\) in Figure \ref{ex1} represent \(P\) using the least number of edges possible. Hence, using \(P\) as the input, the Sparsest Permutation causal discovery algorithm and its variants \citep{UhlSP, lam} would not be able to uniquely output the Markov equivalence class of the causal DAG. The practitioner then has to work with all possible outputs, represented as a candidate class of non-Markov equivalent CPDAGs in Figure \ref{ex1}. 

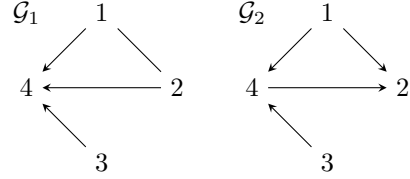
\begin{figure}
    \centering
     \begin{tikzpicture}[>=stealth]
     \node () at (-1,1) {$\mathcal{G}_1$};
\node (e1) at (0,1) {$1$};
\node (e2) at (-1,0) {$4$};
\node (e3) at (1,0) {$2$};
\node (e4) at (0,-1) {$3$};

\draw [->] (e1) to (e2);
\draw [->] (e4) to (e2);
\draw [->] (e3) to (e2);
\draw [-] (e1) to (e3);

\node () at (2,1) {$\mathcal{G}_2$};
\node (e1) at (3,1) {$1$};
\node (e2) at (2,0) {$4$};
\node (e3) at (4,0) {$2$};
\node (e4) at (3,-1) {$3$};

\draw [->] (e1) to (e2);
\draw [->] (e4) to (e2);
\draw [->] (e2) to (e3);
\draw [->] (e1) to (e3);
    \end{tikzpicture}
    \caption{Non-Markov equivalent CPDAGs \(\mathcal{G}_1\) and \(\mathcal{G}_2\) that represent the same distribution \(P\).}\label{ex1}
    \label{fig:enter-label}
\end{figure}
    
In this work, we aim to provide causal effect identification results that are robust when the Markov equivalence class of the causal DAG with background knowledge, cannot be uniquely determined, the setting of which is represented by a candidate set of MPDAGs.

\section{Background}
Slightly abusing notation, capital letters (e.g. \(X\)) will be used to denote both nodes in a graph and the associated variables, and bold capital letters (e.g. \(\boldsymbol{X}\)) will be used to denote both sets of nodes in a graph and the associated random vectors.
\subsection{Graphs}
Let \(\mathcal{G}\) denote a graph over a finite set of nodes \(\boldsymbol{V}\), 
with edges that are either directed  
(\(\rightarrow\)) or   
undirected 
( --- ) between any two adjacent nodes. 
A \emph{path} between the nodes \(V_0\) and \(V_m\) is a sequence of nodes 
\(\langle V_0,\ldots, V_m \rangle\), 
such that for all \(i \in \{0\ldots,m-1\}\), consecutive nodes \(V_i\) and \(V_{i+1}\) are adjacent. If the edges are all directed as \(V_i\xrightarrow{} V_{i+1}\),
the path is a \emph{directed path} from \(V_0\) to \(V_m\); if the edges are either undirected or directed as \(V_i\xrightarrow{} V_{i+1}\), the path is a \emph{semi-directed path}, if, in addition, \(V_0=V_i\), the sequence is a \emph{semi-directed cycle}; if edges between consecutive nodes are all undirected, then the path is an \emph{undirected path}. Given two disjoint subset of nodes \(\boldsymbol{X},\boldsymbol{Y}\subseteq \boldsymbol{V}\), a path from \(\boldsymbol{X}\) to \(\boldsymbol{Y}\) is a path from some \(X\in \boldsymbol{X}\) to some \(Y\in \boldsymbol{Y}\); if only the first node is in \(\boldsymbol{X}\), the path is said to be \emph{proper}.

Given a subset of nodes \(\boldsymbol{D}\subseteq \boldsymbol{V}\) in graph \(\mathcal{G}\), let \(\mathcal{G}_{\boldsymbol{D}}\) denote the induced subgraph, the graph over nodes \(\boldsymbol{D}\), of which the edges are the same as the edges between nodes in \(\boldsymbol{D}\) of graph \(\mathcal{G}\). 

Given a subset of nodes \(\boldsymbol{D}\subseteq \boldsymbol{V}\) in graph \(\mathcal{G}\), let \(\text{Pa}_\mathcal{G}(\boldsymbol{D})\) denote the parents of \(\boldsymbol{D}\) in graph \(\mathcal{G}\), the set of nodes \(X\) such that there exists a directed edge \(X\rightarrow D\) in \(\mathcal{G}\) for some node \(D\in \boldsymbol{D}\), and let \(\text{An}_\mathcal{G}(\boldsymbol{D})\) denote the ancestors of \(\boldsymbol{D}\) in graph \(\mathcal{G}\), the set of nodes \(X\) such that there exists a directed path from \(X\) to some node \(D\in \boldsymbol{D}\) in \(\mathcal{G}\).

A \emph{chain component} \(\boldsymbol{\tau}\) in graph \(\mathcal{G}\) is a maximal set of nodes such that every node \(X,Y\in \boldsymbol{\tau}\) is connected by an undirected path.


A joint
density \(p\) is associated to the set of nodes \(\boldsymbol{V}\) of DAG \(\mathcal{G}\), via the following Markov property for DAGs. 
\begin{definition}[DAG Markov property]
    A density \(p\) is Markovian to DAG \(\mathcal{G}\) if \(p(\boldsymbol{v})=\prod_{V_i\in \boldsymbol{V}}p_{V_i}(v_i\cd \text{pa}_{\mathcal{G}}(v_i))\).
\end{definition}

The \emph{Markov equivalence class} of a DAG \(\mathcal{G}\) is the set of all DAGs that have the same set of d-separations \cite{pearlbook} as the DAG \(\mathcal{G}\), and can be represented by a \emph{completed partially directed acyclic graph} (CPDAG) \citep{meek}.

Given a Markov equivalence class and background knowledge in the form of a set of directed edges \(\mathcal{E}\), we represent a subset of the Markov equivalence class, which consists of DAGs \(\mathcal{G}'\) such that all the edges in \(\mathcal{E}\) are in \(\mathcal{G}'\), as a \emph{maximally oriented partially directed acyclic graph} (MPDAG) \(\mathcal{G}\); let \([\mathcal{G}]\) denote this set of DAGs represented by the MPDAG \(\mathcal{G}\). We say that the joint density \(p\) is \emph{compatible} to MPDAG \(\mathcal{G}\) if \(p\) is Markovian to all DAGs \(\mathcal{G}'\in [\mathcal{G}]\). CPDAGs are then special instances of MPDAGs when \(\mathcal{E}=\emptyset\).


In general, MPDAGs can contain semi-directed cycles, as shown in \citet{perk}, here we will focus on \emph{strictly acyclic}-MPDAGs (SA-MPDAGs) which are MPDAGs with no semi-directed cycles. Since CPDAGs are MPDAGs with no semi-directed cycles, SA-MPDAGs include CPDAGs.

All nodes of an SA-MPDAG \(\mathcal{G}\) can be partitioned into chain components \(\{\boldsymbol{\tau_1},\ldots, \boldsymbol{\tau_m}\}\). 
Given an SA-MPDAG \(\mathcal{G}\) and a node subset \(\boldsymbol{D}\subseteq \boldsymbol{V}\) of \(\mathcal{G}\), let \(CD_{\mathcal{G}}(\boldsymbol{D})\) denote the chain decomposition of \(\boldsymbol{D}\) using the SA-MPDAG \(\mathcal{G}\).
\begin{definition}[Chain Decomposition \(CD_{\mathcal{G}}(.)\)]
    Given an SA-MPDAG \(\mathcal{G}\), and partition of nodes \(\boldsymbol{V}\) into chain components \(\{\boldsymbol{\tau_1},\ldots, \boldsymbol{\tau_m}\}\), the chain decomposition of the node subset \(\boldsymbol{D}\subseteq \boldsymbol{V}\) is the partition \(\{\boldsymbol{D_1}, \ldots, \boldsymbol{D_k}\}\) consisting of the subsets \(\boldsymbol{D}\cap \boldsymbol{\tau_i}\) for \(i\in \{1,\ldots, m\}\) that are non-empty.
\end{definition}

Given \(\boldsymbol{D_i}\in CD_{\mathcal{G}}(\boldsymbol{D})\) for some node subset \(\boldsymbol{D}\subseteq \boldsymbol{V}\) of graph \(\mathcal{G}\), let \(\tau_\mathcal{G}(\boldsymbol{D_i})\) denote the chain component in graph \(\mathcal{G}\) such that \(\boldsymbol{D_i}\subseteq\tau_\mathcal{G}(\boldsymbol{D_i})\).

Note that chain decomposition is similar to the PTO and PCO algorithms from \citet{cpdagidentify} and \citet{perk}, specialised to chain graphs \cite{chain}, of which SA-MPDAGs are a subclass.
\subsection{Causal Effect Identification}
Given an observational density \(p\) Markovian to DAG \(\mathcal{G}\) and node subset \(\boldsymbol{X}\subseteq \boldsymbol{V}\) in \(\mathcal{G}\), we denote \(p^{\textnormal{do}_{\mathcal{G}}(\boldsymbol{x})}\) as the \emph{interventional density} over the remaining nodes \(\boldsymbol{\Bar{V}}=\boldsymbol{V}\backslash \boldsymbol{X}\) after setting variables in \(\boldsymbol{X}\) to a fixed value \(\boldsymbol{x}\), defined as
\begin{align*}
    p^{\textnormal{do}_{\mathcal{G}}(\boldsymbol{x})}(\boldsymbol{\Bar{v}})= \prod_{V_i\in \boldsymbol{\Bar{V}}}p_{V_i}(v_i\cd \textnormal{pa}_{\mathcal{G}}(v_i)),
\end{align*}
for values \(\textnormal{pa}_{\mathcal{G}}(v_i)\) that agree with the intervened values \(\boldsymbol{x}\). 
This is known as the truncated factorisation formula or the g-formula \citep{pearlbook}.

Given disjoint node subsets \(\boldsymbol{X},\boldsymbol{Y}\subseteq\boldsymbol{ V}\), the causal effect of \(\boldsymbol{X}\) on \(\boldsymbol{Y}\), such as the average treatment effect, is a functional of \(p^{\text{do}_{\mathcal{G}}(\boldsymbol{x})}_{\boldsymbol{Y
}}\), the marginal interventional density on \(\boldsymbol{Y}\).

To identify the causal effect of \(\boldsymbol{X}\) on \(\boldsymbol{Y}\),  given that causal DAG \(\mathcal{G}\) is known, \(p^{\text{do}_{\mathcal{G}}(\boldsymbol{x})}_{\boldsymbol{Y
}}\) can be obtained using the observational density \(p\) via the truncated factorisation formula.

However, when the causal DAG is only known up to its Markov equivalence class and some background knowledge, represented as an MPDAG, \citet{perk} provided the following sufficient and necessary criteria when the causal effect is identifiable, that is, when \(p^{\text{do}_{\mathcal{G}}(\boldsymbol{x})}_{\boldsymbol{Y
}}\) is uniquely computable from the DAGs represented by the MPDAG.

\begin{definition}[Identifiability for MPDAGs]
    Given disjoint node subsets \(\boldsymbol{X},\boldsymbol{Y}\subseteq \boldsymbol{V}\) of MPDAG \(\mathcal{G}\), the causal effect of \(\boldsymbol{X}\) on \(\boldsymbol{Y}\) is identifiable if for all observational densities \(p\) compatible to the MPDAG \(\mathcal{G}\), there does not exist DAGs \(\mathcal{G}_1, \mathcal{G}_2\in [\mathcal{G}]\), such that 
    \begin{align*}
        p^{\text{do}_{\mathcal{G}_1}(\boldsymbol{x})}_{\boldsymbol{Y
}}\neq p^{\text{do}_{\mathcal{G}_2}(\boldsymbol{x})}_{\boldsymbol{Y
}}.
    \end{align*}
\end{definition}
If identifiability holds in MPDAG \(\mathcal{G}\), then we can write \(p^{\text{do}_{\mathcal{G}}(\boldsymbol{x})}_{\boldsymbol{Y
}}\) without ambiguity.
\begin{theorem}[\citet{perk}]\label{perkt}
    Given disjoint node subsets \(\boldsymbol{X},\boldsymbol{Y}\subseteq \boldsymbol{V}\) of MPDAG \(\mathcal{G}\), the causal effect of \(\boldsymbol{X}\) on \(\boldsymbol{Y}\) is identifiable if and only if there are no proper semi-directed paths from \(\boldsymbol{X}\) to \(\boldsymbol{Y}\) that start with an undirected edge, with the following identification formula. For any observational density \(p\) compatible to \(\mathcal{G}\),
    \begin{align*}
        p^{\textnormal{do}_{\mathcal{G}}(\boldsymbol{x})}_{\boldsymbol{Y
}}(\boldsymbol{y})=\int \prod^{k}_{j=1}p_{\boldsymbol{B_j}}(\boldsymbol{b_j}\cd {\textnormal{pa}_{\mathcal{G}}(\boldsymbol{B_j})})d\boldsymbol{b},
    \end{align*}
    where \(\boldsymbol{B}=\textnormal{An}_{\mathcal{G}_{\boldsymbol{V\backslash X}}}(\boldsymbol{Y})\backslash \boldsymbol{Y}\), and \(\{\boldsymbol{B_1},\ldots, \boldsymbol{B_{k}}\} = CD_{\mathcal{G}}(\textnormal{An}_{\mathcal{G}_{\boldsymbol{V\backslash X}}}(\boldsymbol{Y}))\), for values \({\textnormal{pa}_{\mathcal{G}}(\boldsymbol{B_j})}\) that agree with the intervened values \(\boldsymbol{x}\).
\end{theorem}
\section{Results}
Given a set of SA-MPDAGs \(\mathbb{G}=\{\mathcal{G}_1,\ldots, \mathcal{G}_n\}\) over the same set of nodes \(\boldsymbol{V}\), we introduce the following.
\begin{definition}[Simultaneous Identifiability]\label{simidentify}
    Given disjoint node subsets \(\boldsymbol{X}\) and \(\boldsymbol{Y}\), and a set of SA-MPDAGs \(\mathbb{G}\), the causal effect of \(\boldsymbol{X}\) on \(\boldsymbol{Y}\) is simultaneously identifiable, if for all observational densities compatible to all \(\mathcal{G}\in \mathbb{G}\), there does not exist DAGs \(\mathcal{G}_1, \mathcal{G}_2\in \bigcup_{\mathcal{G}\in \mathbb{G}}[\mathcal{G}]\) , such that 
    \begin{align*}
       p^{\text{do}_{\mathcal{G}_1}(\boldsymbol{x})}_{\boldsymbol{Y
}}\neq p^{\text{do}_{\mathcal{G}_2}(\boldsymbol{x})}_{\boldsymbol{Y
}}.
    \end{align*}
\end{definition}

Definition \ref{simidentify} states that, given an observational density \(p\) that is compatible to \emph{all} SA-MPDAGs in \(\mathbb{G}\), the interventional marginal density over \(\boldsymbol{Y}\) does not depend on the choice of causal DAG selected from the DAGs represented by all SA-MPDAGs in \(\mathbb{G}\). If simultaneous identifiability holds for \(\mathbb{G}\), we can write \(p^{\text{do}_{\mathbb{G}}(\boldsymbol{x})}_{\boldsymbol{Y
}}\) without ambiguity.

Given disjoint node subsets \(\boldsymbol{X},\boldsymbol{Y}\subseteq \boldsymbol{V}\), we first define Algorithm \ref{rwgraph} to re-weight and marginalise graphs.
\begin{algorithm}[h]
   \caption{Re-weighting and marginalising graphs}\label{rwgraph}
   {\bfseries Input:} graph $\mathcal{G}$, node subset $\boldsymbol{X}, \boldsymbol{Y}\subseteq \boldsymbol{V}$
   
 {\bfseries Output:} graph \( RM(\mathcal{G}; \boldsymbol{X}, \boldsymbol{Y})\)
\begin{algorithmic}[1]
\STATE for all \(I\in \boldsymbol{X}\cap \text{An}_{\mathcal{G}}(\boldsymbol{Y})\), remove directed edges of the form \(J\rightarrow I\), call this graph \(\mathcal{G}(\not \rightarrow \boldsymbol{X})\).

\STATE \textbf{return} \( RM(\mathcal{G}; \boldsymbol{X}, \boldsymbol{Y})\) as the induced subgraph \(\mathcal{G}(\not \rightarrow \boldsymbol{X})_{\text{An}_{\mathcal{G}(\not \rightarrow \boldsymbol{X})}(\boldsymbol{Y})}\).
\end{algorithmic}
\end{algorithm}

Note that step 1 of Algorithm \ref{rwgraph} is just the graphical intervention on \(\mathcal{G}\), except \(\mathcal{G}\) is allowed to be an SA-MPDAG. Step 2 then takes the ancestral margin of \(\boldsymbol{Y}\) of the resulting graph.

\begin{theorem}[Causal Identification Criterion]\label{main}
Given disjoint node subsets \(\boldsymbol{X}\) and \(\boldsymbol{Y}\) and a set of SA-MPDAGs \(\mathbb{G}\), let \(\boldsymbol{A_i}= \boldsymbol{X}\cap \textnormal{An}_{\mathcal{G}_i}(\boldsymbol{Y})\). If
\begin{enumerate}
    \item for all \(\mathcal{G}_i\in \mathbb{G}\), there are no proper semi-directed paths from \(\boldsymbol{X}\) to \(\boldsymbol{Y}\) that start with an undirected edge, and
    \item for any pair \(\mathcal{G}_i, \mathcal{G}_j \in \mathbb{G}\), one of the following holds:
    \begin{enumerate}
        \item for all  \(\boldsymbol{X_{k}} \in CD_{\mathcal{G}_i}( \boldsymbol{A_i})\), 
        \(\textnormal{Pa}_{\mathcal{G}_i}(\boldsymbol{X_{k}})= \textnormal{Pa}_{\mathcal{G}_j}(\boldsymbol{X_{k}})\) and likewise for the roles of \(i\) and \(j\) reversed, or
    \item graphs \(RM(\mathcal{G}_i;\boldsymbol{X},\boldsymbol{Y})\) and \(RM(\mathcal{G}_j;\boldsymbol{X},\boldsymbol{Y})\) are Markov equivalent,
    \end{enumerate}
\end{enumerate}

    then the causal effect of \(\boldsymbol{X}\) on \(\boldsymbol{Y}\) is simultaneously identifiable with the following identification formula. For any observational density \(p\) compatible to any \(\mathcal{G}_i\in \mathbb{G}\), 
    \begin{align*}
        p^{\text{do}_{\mathbb{G}}(\boldsymbol{x})}_{\boldsymbol{Y
}}(\boldsymbol{y})=\int \prod^{k}_{j=1}p_{\boldsymbol{B_j}}(\boldsymbol{b_j}\cd {\textnormal{pa}_{\mathcal{G}_i}(\boldsymbol{B_j})})d\boldsymbol{b}
    \end{align*}
    for any \(\mathcal{G}_i\in \mathbb{G}\), where \(\boldsymbol{B}=\textnormal{An}_{(\mathcal{G}_i)_{\boldsymbol{V\backslash X}}}(\boldsymbol{Y})\backslash \boldsymbol{Y}\), and \(\{\boldsymbol{B_1},\ldots, \boldsymbol{B_{k}}\} = CD_{\mathcal{G}_i}(\textnormal{An}_{(\mathcal{G}_i)_{\boldsymbol{V\backslash X}}}(\boldsymbol{Y}))\), for values \({\textnormal{pa}_{\mathcal{G}_i}(\boldsymbol{B_j})}\) that agree with the intervened values \(\boldsymbol{x}\).
\end{theorem}

Note that simultaneously identifiability in \(\mathbb{G}\) implies identifiability in each \(\mathcal{G}_i\in \mathbb{G}\) which, by Theorem \ref{perkt}, is equivalent to condition 1 in Theorem \ref{main}, thus condition 1 is a necessary condition for simultaneous identifiability. The identification formula in Theorem \ref{main} is the identification formula in \citet{perk} applied to some \(\mathcal{G}_i\in \mathbb{G}\).

In the case of \(\mathcal{G}_i\) and \(\mathcal{G}_j\) being DAGs, the interventional densities \(p^{\text{do}_{\mathcal{G}_i}(\boldsymbol{x})}\) and \(p^{\text{do}_{\mathcal{G}_j}(\boldsymbol{x})}\) can be obtained from observational density \(p\) by re-weighting based on \(\mathcal{G}_i\) and \(\mathcal{G}_j\) respectively.  Condition 2a in Theorem \ref{main} can be thought of as a graphical condition ensuring that this re-weighting step is the same for both \(\mathcal{G}_i\) and \(\mathcal{G}_j\), generalised to SA-MPDAGs.

The Markov equivalence in condition 2b is with respect to the Markov property defined for chain graphs, characterised in \citet{chain}. Condition 2b can be thought of as the ancestral subgraphs of SA-MPDAGs \(\mathcal{G}_i\) and \(\mathcal{G}_j\) being interventionally Markov equivalent, however, this differs from \citet{I-MEC} in that observational Markov equivalence is not required and that Markov equivalence need only hold on the ancestral subgraph.
\subsection{Examples}
We illustrate Theorem \ref{main} using some examples. 

\begin{example}[Case when condition 2a holds]
Consider again the CPDAGs in Figure \ref{ex1}, we are interested in the causal effect of \(3\) on \(2\). 

There are no proper semi-directed paths from \(3\) to \(2\) that start with an undirected edge in both CPDAGs, thus condition 1 of Theorem \ref{main} holds. Since \(A_1=\emptyset\) and \(A_2=\{3\}\), for \(\{3\} \in CD_{\mathcal{G}_2}(A_2)\), we have \(\text{pa}_{\mathcal{G}_1}(\{3\})= \text{pa}_{\mathcal{G}_2}(\{3\})= \emptyset\). Thus, condition 2a of Theorem \ref{main} holds; given \(\mathbb{G}=\{\mathcal{G}_1, \mathcal{G}_2\}\), the causal effect of \(3\) on \(2\) is simultaneously identifiable.

In this case, intervening on \(3\) is equivalent to conditioning on \(3\). Note that \(p\) being compatible to CPDAG \(\mathcal{G}_1\) implies the independence \(2\ci 3\), as such although there exists a directed path from \(3\) to \(2\) in \(\mathcal{G}_2\), the marginal distribution of \(2\) is not affected by intervening on \(3\). Note that this example holds for any distribution \(P\) satisfying \(X_1\ci X_2\) and \(X_1 \ci X_2 \cd \{X_3, X_4\}\) and \(X_3\ci X_4 \cd \{X_1, X_2\}\).
\end{example}

\begin{example}[Case when condition 2b holds]
    \begin{figure}[h]
    \centering
     \begin{tikzpicture}[>=stealth]
     \node () at (-1.5,1) {$\mathcal{G}_1$};
\node (e1) at (0,1) {$1$};
\node (e2) at (-1,0) {$4$};
\node (e3) at (1,0) {$2$};
\node (e4) at (0,-1) {$3$};
\node (e5) at (-2,-1) {$5$};

\draw [->] (e1) to (e2);
\draw [->] (e4) to (e2);
\draw [-] (e1) to (e3);
\draw [-] (e4) to (e3);
\draw [->] (e2) to (e5);
\draw [->, bend right] (e1) to (e5);

\node () at (2.5,1) {$\mathcal{G}_2$};
\node (e1) at (4,1) {$1$};
\node (e2) at (3,0) {$4$};
\node (e3) at (5,0) {$2$};
\node (e4) at (4,-1) {$3$};
\node (e5) at (2,-1) {$5$};

\draw [->] (e1) to (e2);
\draw [->] (e1) to (e3);
\draw [->] (e4) to (e3);
\draw [->] (e2) to (e5);
\draw [->, bend right] (e1) to (e5);

\node () at (1.5,-2) {\small $RM(\mathcal{G}_1; \{4\}, \{5\})= RM(\mathcal{G}_2; \{4\}, \{5\})$};
\node (e1) at (2.5,-3) {$1$};
\node (e2) at (1.5,-3) {$4$};
\node (e5) at (0.5,-3) {$5$};

\draw [->] (e2) to (e5);
\draw [->, bend right] (e1) to (e5);
    \end{tikzpicture}
    \caption{Non-Markov equivalent SA-MPDAGs \(\mathcal{G}_1\) and \(\mathcal{G}_2\) with different adjacencies.}\label{ex2}
\end{figure}
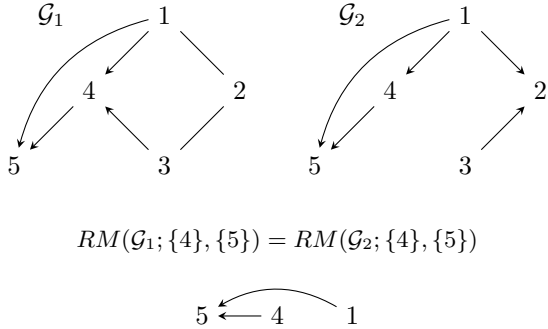
Consider the SA-MPDAGs in Figure \ref{ex2}, we are interested in the causal effect of \(4\) on \(5\). 

There are no proper semi-directed paths from \(4\) to \(5\) that start with an undirected edge in both CPDAGs, thus condition 1 of Theorem \ref{main} holds. We have \(RM(\mathcal{G}_1; \{4\}, \{5\})=RM(\mathcal{G}_2; \{4\}, \{5\})\). Thus, condition 2b of Theorem \ref{main} holds; given \(\mathbb{G}=\{\mathcal{G}_1, \mathcal{G}_2\}\), the causal effect of \(4\) on \(5\) is simultaneously identifiable.

Note that SA-MPDAGs of different adjacencies such as those in Figure \ref{ex2} may be obtained when the skeleton building step of the PC algorithm, which is the most computationally intensive part of the PC algorithm, lacks power.
\end{example}

\section{Discussion and Future Directions}
Causal effect identification when the causal graph cannot be fully specified is not new, with many settings for doing so being proposed. In our setting, the Markov equivalence class of the true causal DAG \(\mathcal{G}_0\) cannot be uniquely determined and we can only work with a candidate set of SA-MPDAGs (which includes CPDAGs) \(\mathbb{G}=\{\mathcal{G}_1,...,\mathcal{G}_n\}\) such that \(\mathcal{G}_0\in \bigcup_{\mathcal{G}_i\in \mathbb{G}}[\mathcal{G}_i]\), representing the Markov equivalence classes with background knowledge which could contain \(\mathcal{G}_0\). This can happen when untestable assumptions (e.g. faithfulness, Sparsest Markov Representation (SMR)) are not satisfied, causing causal discovery algorithms to not return a unique Markov equivalence class such as in Figure \ref{ex1}, or when the conditional independence testing step of causal discovery lacks power in detecting edges such as in Example \ref{ex2}. 

\subsection{Definition of Simultaneous Identifiability}
In Definition \ref{simidentify}, observational density \(p\) is Markovian to \emph{all} DAGs in \(\bigcup_{\mathcal{G}\in \mathbb{G}}[\mathcal{G}]\). When the represented DAGs are not Markov equivalent, a natural question would be: ``
Can density \(p\) be \emph{uniquely} Markovian to some graph that is not necessarily a DAG?''. \citet{teh1} has shown that even for anterial graphs \citep{lkayvan}, a general class of graphs containing DAGs, chain graphs and ancestral graphs, there still exist observational distributions that cannot be uniquely Markovian to \emph{any} anterial graph.


\subsection{Necessity and Restricting \(\mathbb{G}\)}
Theorem \ref{main} applies to any candidate set \(\mathbb{G}\)--no assumptions about \(\mathbb{G}\) has been made. 

To show the necessity of Theorem \ref{main}, given a set of SA-MPDAGs \(\mathbb{G}\), a counter-example of the following form is needed: a density \(p\) that is Markovian to all SA-MPDAGs in \(\mathbb{G}\) such that \(\mathbb{G}\) does not satisfy the conditions in Theorem \ref{main} and \(p^{\text{do}_{\mathcal{G}}(\boldsymbol{x})}_{\boldsymbol{Y
}}\neq p^{\text{do}_{\mathcal{G}'}(\boldsymbol{x})}_{\boldsymbol{Y
}}\) for some DAGs \(\mathcal{G}, \mathcal{G}'\in \bigcup_{\mathcal{G}_i\in \mathbb{G}}[\mathcal{G}_i]\).

Without further assumptions on the candidate set \(\mathbb{G}\), it can be challenging to construct, in general, such a density \(p\) to show necessity. 

In practice, also, \(\mathbb{G}\) is usually not arbitrary. For example, consider the SA-MPDAGs in \(\mathbb{G}\) of Figure \ref{ex1}, which are the Markovian CPDAGs to some distribution \(P\) with the least number of edges possible, representing all possible outputs of the Sparsest Permutation algorithm when faithfulness is violated. The SA-MPDAGs in \(\mathbb{G}\) may also have the same background knowledge (see Figure \ref{ex2} with both \(\mathcal{G}_1,\mathcal{G}_2\) having background knowledge \(\mathcal{E}=\{1\rightarrow 4, 4\rightarrow 5\}\)), which represents the case of applying the same background knowledge to different CPDAGs obtained from causal discovery.


\subsection{Generalising to MPDAGs}
The proof of Theorem \ref{main} relies on results for chain graphs in \citet{chain}. Since MPDAGs can have semi-directed cycles, they are in general not chain graphs. Hence, a different proof technique is needed to show simultaneous identifiability when \(\mathbb{G}\) is a set of general MPDAGs.

\nocite{langley00}

\bibliography{example_paper}
\bibliographystyle{icml2025}

\newpage
\appendix
\onecolumn
\section{Proofs}

We will use the following results on MPDAGs and chain graphs.

\begin{proposition}[Properties of SA-MPDAGs]\label{cpdagprop}
    SA-MPDAG \(\mathcal{G}\) contains only directed and undirected edges, and there are no nodes \(I,J,K\) of \(\mathcal{G}\), such that \(I\rightarrow J -K\) and \(I\) is not adjacent to \(K\). 
\end{proposition}

Proposition \ref{cpdagprop} holds for all MPDAGs, and is a straightforward consequence of the orientation rules from \citet{meek}, and implies that in an SA-MPDAG, all nodes sharing a chain component have the same parents.

\begin{proposition}[Markov properties of SA-MPDAGs]\label{cpdagmarkov}
    Given an SA-MPDAG \(\mathcal{G}\), for any DAG \(\mathcal{G}'\in [\mathcal{G}]\), \(\mathcal{G}\) and \(\mathcal{G}'\) are Markov equivalent.
\end{proposition}

In Proposition \ref{cpdagmarkov}, the Markov property for SA-MPDAG \(\mathcal{G}\) refers to the Markov property on chain graphs \citep{chain}, of which SA-MPDAGs are a subclass. Hence, Proposition \ref{cpdagmarkov} implies that if a joint density \(p\) is compatible to SA-MPDAG \(\mathcal{G}\), then \(p\) is Markovian to SA-MPDAG \(\mathcal{G}\), admitting a chain graph factorisation using \(\mathcal{G}\).

\begin{proof}[Proof of Proposition \ref{cpdagmarkov}]

     We use the following from \citet{chain}:
    \begin{lemma}\label{lemmarkov}
        Two chain graphs \(\mathcal{G}'\) and \(\mathcal{G}\) are Markov equivalent if they have the same skeleton and minimal complexes.
    \end{lemma}
    
    A minimal complex in chain graph \(\mathcal{G}\) is
    a triple  \((I,\boldsymbol{C},J)\) such that node \(I\) is not adjacent to node \(J\), and \(\boldsymbol{C}\) is the minimal subset of a chain component such that \(I\rightarrow C_1\) and \(J\rightarrow C_2\) for some \(C_1,C_2\in \boldsymbol{C}\). We will show that all minimal complexes in SA-MPDAG \(\mathcal{G}\) are unshielded colliders. Since in SA-MPDAG \(\mathcal{G}\),  \(I\rightarrow J-K\)  implies that \(I\) is adjacent to \(K\), and we must have \(I\rightarrow K\) otherwise the semi-directed cycle \(I\rightarrow J- K \rightarrow I\) or \(I\rightarrow J- K - I\) will occur in SA-MPDAG \(\mathcal{G}\), a contradiction. Thus, all minimal complexes in \(\mathcal{G}\) are of of the form \((I,K,J)\), with \(K\) a single node and is thus a collider.

    By construction, the colliders in  represented DAG \(\mathcal{G}'\) and SA-MPDAG \(\mathcal{G}\) must coincide and share the same skeleton, and since DAGs are chain graphs, by Lemma \ref{lemmarkov}, DAG \(\mathcal{G'}\) is Markov equivalent to SA-MPDAG \(\mathcal{G}\). 
\end{proof}

We also have the following observation on output graphs of Algorithm \ref{rwgraph}.
\begin{lemma}
    If there are no proper semi-directed paths from \(\boldsymbol{X}\) to \(\boldsymbol{Y}\) in SA-MPDAG \(\mathcal{G}\) that start with an undirected edge, then in \(RM(\mathcal{G}; \boldsymbol{X}, \boldsymbol{Y})\), \(I\rightarrow J-K\) implies \(I\rightarrow K\).
\end{lemma}
\begin{proof}
    Observe that chain components \(\boldsymbol{\tau}\) with directed edges into \(\boldsymbol{\tau}\) being removed in \(RM(\mathcal{G}; \boldsymbol{X}, \boldsymbol{Y})\) are the chain components \(\boldsymbol{\tau}\) such that \(\boldsymbol{\tau}\cap \boldsymbol{X}\neq \emptyset\), for which we have \(\boldsymbol{\tau}\backslash \boldsymbol{X}\cap \text{An}_{\mathcal{G}(\not \rightarrow \boldsymbol{X})}(\boldsymbol{Y})=\emptyset\), otherwise there exists a proper semi-directed path from \(\boldsymbol{X}\) to \(\boldsymbol{Y}\) that starts with undirected edges. Thus, chain components \(\boldsymbol{\tau}\) with directed edges into \(\boldsymbol{\tau}\) being removed does not appear as nodes in the induced subgraph \(\mathcal{G}(\not \rightarrow \boldsymbol{X})_{\text{An}_{\mathcal{G}(\not \rightarrow \boldsymbol{X})}(\boldsymbol{Y})}\), and thus \(RM(\mathcal{G}; \boldsymbol{X}, \boldsymbol{Y})\). 
\end{proof}

Observe that the proof of Proposition \ref{cpdagmarkov} only uses the fact that \(I\rightarrow J-K\) implies that \(I\rightarrow K\) in graph \(\mathcal{G}\), and \(\mathcal{G}'\) and \(\mathcal{G}\) share skeletons and unshielded colliders. This property is also satisfied by the graph \(RM(\mathcal{G}; \boldsymbol{X},\boldsymbol{Y})\) obtained from input SA-MPDAG \(\mathcal{G}\) and node subsets \(\boldsymbol{X}\) and \(\boldsymbol{Y}\) that satisfy the condition of Theorem \ref{perkt} in graph \(\mathcal{G}\). Thus, analogously \(RM(\mathcal{G}; \boldsymbol{X},\boldsymbol{Y})\) is Markov equivalent to some DAG \(\mathcal{G}\).

We have the following proposition.
\begin{proposition}\label{equal factorise}
    For SA-MPDAGs \(\mathcal{G}_1\) and \(\mathcal{G}_2\) and node subsets \(\boldsymbol{X}\) and \(\boldsymbol{Y}\) that satisfies the condition of Theorem \ref{perkt} in both \(\mathcal{G}_1,\mathcal{G}_2\), if \(RM(\mathcal{G}_1; \boldsymbol{X},\boldsymbol{Y})\) is Markov equivalent to \(RM(\mathcal{G}_2; \boldsymbol{X},\boldsymbol{Y})\), then the factorisation of density \(p\) based on either graph has the same evaluation.
\end{proposition}

\begin{proof}[Proof of Proposition \ref{equal factorise}]
 \(RM(\mathcal{G}_1; \boldsymbol{X},\boldsymbol{Y})\), \(RM(\mathcal{G}_2; \boldsymbol{X},\boldsymbol{Y})\), and respective DAGs that \(RM(\mathcal{G}_1; \boldsymbol{X},\boldsymbol{Y})\), \(RM(\mathcal{G}_2; \boldsymbol{X},\boldsymbol{Y})\) are Markov equivalent to, \(\mathcal{G}'_1\) and \(\mathcal{G}'_2\), are all Markov equivalent. 

Since there always exists a sequence of covered edge flips that transforms between Markov equivalent DAGs \cite{ceflips}, the  factorisation of the density \(p\) based on DAG \(\mathcal{G}'_1\) can be transformed into the factorisation of 
 the same density \(p\) based on DAG \(\mathcal{G}'_2\) without changing the evaluation. By Proposition \ref{cpdagmarkov}, the factorisation of \(p\) based on DAG \(\mathcal{G}'_1\) and the factorisation of \(p\) based on \(RM(\mathcal{G}_1; \boldsymbol{X},\boldsymbol{Y})\) have the same evaluation, likewise for DAG \(\mathcal{G}_2\) and \(RM(\mathcal{G}_2; \boldsymbol{X},\boldsymbol{Y})\). Thus, the factorisation of density \(p\) based on either \(RM(\mathcal{G}_1; \boldsymbol{X},\boldsymbol{Y})\) or \(RM(\mathcal{G}_2; \boldsymbol{X},\boldsymbol{Y})\) has the same evaluation.
\end{proof}

We also have the following proposition re-expressing the identification formula in Theorem \ref{perkt} from \cite{perk}, as a marginal distribution of a re-weighted distribution of the observational density \(p\). In the case when \(\mathcal{G}\) is a DAG, both sides reduces to the truncated factorisation formula \cite{pearlbook}.
\begin{proposition}\label{conversion}
    For variables \(\boldsymbol{X}\) (with intervened value \(\boldsymbol{x}\)) and variables \(\boldsymbol{Y}\) in SA-MPDAG \(\mathcal{G}\) such that the condition of Theorem \ref{perkt} holds, i.e. there does not exist proper semi-directed paths that start with an undirected edge from \(X\) to \(Y\) in \(\mathcal{G}\), and observational density \(p\) compatible to \(\mathcal{G}\),
    \begin{equation}
    \int \prod^{k}_{j=1}p_{\boldsymbol{B}_j}(\boldsymbol{b_j}\cd {\textnormal{pa}_{\mathcal{G}}(\boldsymbol{B_j})})d\boldsymbol{b}=\int \frac{p(\boldsymbol{v})}{\prod^{m}_{j=1} p_{\boldsymbol{X_j}}(\boldsymbol{x_j}\cd \boldsymbol{\textnormal{pa}_{\mathcal{G}}(X_j))}} d\boldsymbol{v'}, \label{reweighting}
\end{equation}
where \(\{\boldsymbol{B_1},\ldots, \boldsymbol{B_{k}}\} = CD_{\mathcal{G}}(\textnormal{An}_{\mathcal{G}_{\boldsymbol{V\backslash X}}}(\boldsymbol{Y}))\),  \(\boldsymbol{B}=\textnormal{An}_{\mathcal{G}_{\boldsymbol{V\backslash X}}}(\boldsymbol{Y})\backslash \boldsymbol{Y}\), \(\{\boldsymbol{X_1},...,\boldsymbol{X_{m}}\}= CD_{\mathcal{G}}(\boldsymbol{X}\cap \textnormal{An}_{\mathcal{G}}(\boldsymbol{Y}))\), and \(\boldsymbol{V}'=\boldsymbol{V\backslash(X\cup Y)}\), for values \({\textnormal{pa}_{\mathcal{G}}(\boldsymbol{B_j})}\) and \({\textnormal{pa}_{\mathcal{G}}(\boldsymbol{X_j})}\) that agree with the intervened values \(\boldsymbol{x}\).
\end{proposition}
\begin{proof}[Proof of Proposition \ref{conversion}]
We will start with the RHS, showing that it is equal to the LHS in \ref{reweighting}. By Proposition \ref{cpdagmarkov}, observational density \(p\) is Markovian to SA-MPDAG \(\mathcal{G}\) and we have the following chain graph factorisation
\begin{equation}
    p(\boldsymbol{v})=\prod^{\ell}_{j=1} p_{\boldsymbol{c_j}}( \boldsymbol{c_j}\cd \text{pa}_{\mathcal{G}}(\boldsymbol{C_j})),\label{jointchainfactor}
\end{equation}
where \(\boldsymbol{C_1},\ldots,\boldsymbol{C_{\ell}}\) are the chain components of SA-MPDAG \(\mathcal{G}\).  

 Observe that for any \(j_1,j_2\), we have that \(\boldsymbol{B_{j_1}}\) and \(\boldsymbol{X_{j_2}}\) cannot be from the same chain component, otherwise there exists an undirected path from \(\boldsymbol{X}\) to \(\text{An}_{\mathcal{G}_{\boldsymbol{V\backslash X}}}(\boldsymbol{Y})\), implying a proper semi-directed path that starts with undirected edges from \(\boldsymbol{X}\) to \(\boldsymbol{Y}\) in \(\mathcal{G}\), a contradiction. Thus by Bayes rule and \ref{jointchainfactor}, we can factorise the integrand of the RHS in \ref{reweighting}, as
\begin{equation}\label{eq3}
\prod^{\ell_1}_{j_1=1} p_{\boldsymbol{C_{j_1}}}( \boldsymbol{c_{j_1}}\cd \text{pa}_{\mathcal{G}}(\boldsymbol{C_{j_1}}))\prod^{\ell_2}_{j_2=1} p_{\boldsymbol{C'_{j_2}}}( \boldsymbol{c'_{j_2}}\cd \text{pa}_{\mathcal{G}}(\boldsymbol{C_{j_2}}), \boldsymbol{x_{j_2}}),
\end{equation}
where \(\boldsymbol{C'_{j_2}}=\boldsymbol{C_{j_2}\backslash X_{j_2}}\), and for values \({\textnormal{pa}_{\mathcal{G}}(\boldsymbol{C_{j_1}})}\) and \({\textnormal{pa}_{\mathcal{G}}(\boldsymbol{C_{j_2}})}\) that agree with the intervened values \(\boldsymbol{x}\).
Note that \(\boldsymbol{C'_{j_2}}\) does not intersect \(\text{An}_{\mathcal{G}_{\boldsymbol{V\backslash X}}}(\boldsymbol{Y})\) otherwise there exists a proper semi-directed path that starts with undirected edges from \(\boldsymbol{X}\) to \(\boldsymbol{Y}\) in \(\mathcal{G}\), a contradiction. Hence, when computing the marginal density of variables \(\boldsymbol{Y}\) using \ref{eq3}, the second product in \ref{eq3} sums to 1,
by summing over each \(\boldsymbol{C'_{j_2}}\) from the last chain component in \(\mathcal{G}_{\boldsymbol{V\backslash X}}\). Likewise the chain components \(\boldsymbol{C_{j_1}}\) such that \(\tau_{j_1}\) does not intersect \(\text{An}_{\mathcal{G}_{\boldsymbol{V\backslash X}}}(\boldsymbol{Y})\) can be summed out. Thus marginalising \ref{eq3} is equivalent to marginalising the expression
\begin{equation}\label{eq4}
    \prod^{\ell'_1}_{j_1=1} p_{\boldsymbol{C_{j_1}}}( \boldsymbol{c_{j_1}}\cd \text{pa}_{\mathcal{G}}(\boldsymbol{C_{j_1}})),
\end{equation}
where \(\boldsymbol{C_1}, \ldots, \boldsymbol{C_{\ell'_1}}\) are the \(\boldsymbol{C_{j_1}}\) in the first product of \ref{eq3} such that \(\boldsymbol{C_{j_1}}\cap \text{An}_{\mathcal{G}_{\boldsymbol{V\backslash X}}}(\boldsymbol{Y})\neq \emptyset\), and for values \({\textnormal{pa}_{\mathcal{G}}(\boldsymbol{C_{j_1}})}\) that agree with the intervened values \(\boldsymbol{x}\).


There do not exist directed paths from \(I\in \boldsymbol{C_{j_1}}\backslash \text{An}_{\mathcal{G}_{\boldsymbol{V\backslash X}}}(\boldsymbol{Y})\) to any of the other \(\text{Pa}_{\mathcal{G}}(\boldsymbol{C_i})\) in \ref{eq4}, otherwise by Proposition \ref{cpdagprop}, \(\text{Pa}_{\mathcal{G}}(\boldsymbol{C_i})= \text{Pa}_{\mathcal{G}}(\boldsymbol{B_i})\), for some \(\boldsymbol{B_i}\) such that \(\boldsymbol{C_i}=\tau_{\mathcal{G}}(\boldsymbol{B_i})\), and we have \(I\in \text{An}_{\mathcal{G}_{\boldsymbol{V\backslash X}}}(\boldsymbol{Y})\), a contradiction. 

Thus, we can sum over the nodes in \(\boldsymbol{C_{j_1}}\backslash \text{An}_{\mathcal{G}_{\boldsymbol{V\backslash X}}}(\boldsymbol{Y})= \boldsymbol{C_{j_1}}\backslash \boldsymbol{B_{j_1}}\), and replace  \(\text{Pa}_{\mathcal{G}}(\boldsymbol{C_{j_1}})= \text{Pa}_{\mathcal{G}}(\boldsymbol{B_{j_1}})\) using Proposition \ref{cpdagprop} in \ref{eq4} gives the LHS of \ref{reweighting}. 
\end{proof}

\begin{proof}[Proof of Theorem \ref{main}]
Applying Theorem \ref{perkt} to each SA-MPDAG \(\mathcal{G}_i\in \mathbb{G}\),  \(p^{\text{do}_{\mathcal{G}_i}(\boldsymbol{x})}_{\boldsymbol{Y
}}\) is identifiable in each \(\mathcal{G}_i\) if and only if condition 1 of Theorem \ref{main} holds, and can be expressed via \ref{reweighting} of Proposition \ref{conversion} for graph SA-MPDAG \(\mathcal{G}_i\) as:
\begin{equation}\label{converted}
    \int \frac{p(\boldsymbol{v})}{\prod^{m}_{j=1} p_{\boldsymbol{X_j}}(\boldsymbol{x_j}\cd \boldsymbol{\textnormal{pa}_{\mathcal{G}_i}(X_j))}} d\boldsymbol{v'}
\end{equation}
   \(\{\boldsymbol{X_1},...,\boldsymbol{X_{m}}\}= CD_{\mathcal{G}_i}(\boldsymbol{X}\cap \text{An}_{\mathcal{G}_i}(\boldsymbol{Y}))\), \(\boldsymbol{V}'=\boldsymbol{V\backslash(X\cup Y)}\), and
   for values \({\textnormal{pa}_{\mathcal{G}_i}(\boldsymbol{X_j})}\) that agree with the intervened values \(\boldsymbol{x}\).
We will show that if condition 2a from Theorem \ref{main} holds for both \(\mathcal{G}_i\) and \(\mathcal{G}_j\), then \ref{converted} will have the same evaluation. 

Note that the LHS of \ref{reweighting} (the identification formula) in Proposition \ref{conversion} for graph \(\mathcal{G}_i\) only contains terms with variables from \(\text{An}_{(\mathcal{G}_i)_{\boldsymbol{V\backslash X}}}(\boldsymbol{Y})\), thus the evaluation of \ref{converted} is not changed by introducing denominator terms of the form \(p_{\boldsymbol{X_\ell}}(\boldsymbol{x_{\ell}}\cd \text{pa}_{\mathcal{G}_i}(\boldsymbol{X_{\ell}}))\), where \(\boldsymbol{X_{\ell}}\) is a subset of some chain component in \(\mathcal{G}_i\), such that \(\boldsymbol{X}_\ell\cap \text{An}_{(\mathcal{G}_i)_{\boldsymbol{V\backslash X}}}(\boldsymbol{Y})=\emptyset\) to the integrand in  \ref{converted}.

Given SA-MPDAGs \(\mathcal{G}_{i}, \mathcal{G}_{j}\in \mathbb{G}\). If condition 2a in Theorem \ref{main} iholds, for \(\boldsymbol{X_{\ell}}\in CD_{\mathcal{G}_{j}}(\boldsymbol{A_{j}\backslash A_{i}})\), we have 
\(p_{\boldsymbol{X}_\ell}(\boldsymbol{x_{\ell}}\cd \text{pa}_{\mathcal{G}_{j}}(\boldsymbol{X_{\ell}}))= p_{\boldsymbol{X}_\ell}(\boldsymbol{x_{\ell}}\cd \text{pa}_{\mathcal{G}_{i}}(\boldsymbol{X_{\ell}}))\).
 Hence, note that we can introduce denominators of the form \(p_{\boldsymbol{X}_\ell}(\boldsymbol{x_{\ell}}\cd \text{pa}_{\mathcal{G}_{j}}(\boldsymbol{X_{\ell}}))\) from the integrand in \ref{converted} for graph \(\mathcal{G}_j\), where \(\boldsymbol{X_{\ell}}\in CD_{\mathcal{G}_{j}}(\boldsymbol{A_{j}\backslash A_{i}})\), as the equivalent term \(p_{\boldsymbol{X}_\ell}(\boldsymbol{x_{\ell}}\cd \text{pa}_{\mathcal{G}_{i}}(\boldsymbol{X_{\ell}}))\) to the integrand in \ref{converted} for graph \(\mathcal{G}_{i}\) without changing the evaluation of \ref{converted} for graph \(\mathcal{G}_i\). Likewise with the roles of \(\mathcal{G}_i\) and \(\mathcal{G}_j\) reversed. We can repeat this procedure until the denominator in the integrand in \(\ref{converted}\) is the same for both \(\mathcal{G}_i\) and \(\mathcal{G}_j\) and since the observational density \(p\) is the same for both \(\mathcal{G}_i\) and \(\mathcal{G}_j\), \ref{converted} will have the same evaluation for both \(\mathcal{G}_i\) and \(\mathcal{G}_j\). 

We will show that if condition 2b from Theorem \ref{main} holds for both \(\mathcal{G}_i\) and \(\mathcal{G}_j\), then the LHS of \ref{reweighting} in Proposition \ref{conversion} (identification formula) will have the same evaluation.
Observe that the integrand density of the identification formula for graph \(\mathcal{G}_i\) factorises based on \(RM(\mathcal{G}_i;\boldsymbol{X},\boldsymbol{Y})\). If condition 2b in Theorem \ref{main} holds for SA-MPDAGs \(\mathcal{G}_{i}, \mathcal{G}_{j}\in \mathbb{G}\), 
 by Proposition \ref{equal factorise}, the integrand density in the identification formula for both \(\mathcal{G}_i\) and \(\mathcal{G}_j\) have the same evaluation.
\end{proof}

\end{document}

%% file: math.tex
\usepackage{mathtools}
\usepackage{amsmath}
\usepackage{amssymb}
\usepackage{amsthm}

\newcommand{\ci}{\mbox{\protect $\: \perp \hspace{-2.3ex}
\perp$ }}

\newcommand{\cd}{\,|\,}

\theoremstyle{plain}
\theoremstyle{definition}
\theoremstyle{remark}
\newtheorem{example}{Example}